\documentclass[letterpaper, 10 pt, conference]{ieeeconf}  

\IEEEoverridecommandlockouts                              
\usepackage{amsfonts}
\usepackage{caption}
\usepackage{subcaption}

\usepackage{bm}
\usepackage{graphicx}
\usepackage{graphics} 
\usepackage{epsfig} 
\usepackage[outdir=./]{epstopdf}
\usepackage{mathptmx} 
\usepackage{times} 
\usepackage{amsmath} 
\usepackage{amssymb}  
\usepackage{color}
\usepackage{array}
\usepackage{floatflt}
\usepackage{lipsum}
\usepackage{amsfonts} 
\usepackage{lastpage}
\usepackage{latexsym}
\usepackage{booktabs}
 \usepackage{algorithm}
\usepackage[noend]{algorithmic}

\newcommand{\LTLEVENTUALLY}{\ensuremath{ F}}
\newcommand{\LTLALWAYS}{\ensuremath{ G }}

\newcommand{\FUNCTION}[2]{\STATE \textbf{function} #1(#2)}
\newcommand{\st}{\sigma_{\tau}}
\newcommand{\pt}{\psi}
\newcommand{\St}{\mathcal{S}_{\tau}}
\newcommand{\Mt}{\mathcal{M}_{\tau}}
\newcommand{\sigt}{s^{t:t+\tau}}
\newtheorem{definition}{Definition}
\newtheorem{assump}{Assumption}

\newtheorem{example}{Example}

\newtheorem{proposition}{\bf{Proposition}}

\newtheorem{problem}{Problem}
\graphicspath{{Figures/}{../CISEDay2015/Figures/}}
\newcommand{\tuple}[1]{\ensuremath{\left \langle #1 \right \rangle }}

\title{\LARGE \bf
Robust Satisfaction of Temporal Logic Specifications via Reinforcement Learning
}

\author{Austin Jones$^{1}$, Derya Aksaray$^{2}$, Zhaodan Kong$^{3}$,
  Mac Schwager$^{4}$, and Calin Belta$^{2,5}$ 
  \thanks{*This work was partially supported at Boston University by
    ONR grant number N00014-14-1-0554 and by
    the NSF under grant numbers NRI-1426907, CMMI-1400167, and CNS-1035588.}
    \thanks{$^{1}$Author is with Mechanical Engineering and Electrical Engineering,
    Georgia Institute of Technology, Atlanta, GA, USA. austinjones@gatech.edu}%
  \thanks{$^{2}$Authors are with Mechanical Engineering, Boston
    University, Boston, MA, USA. \{cbelta,daksaray\}@bu.edu}%
  \thanks{$^{3}$Author is with Mechanical and Aerospace Engineering,
    University of California Davis, Davis, CA,
    USA. zdkong@ucdavis.edu}%
     \thanks{$^{4}$ Author is with  Aeronautics and Astronautics, Stanford University,
    Stanford, CA, USA. schwager@stanford.edu}%
    \thanks{$^{5}$Author is with Systems Engineering, Boston
    University, Boston, MA, USA.}%
}

\begin{document}
\maketitle
\thispagestyle{empty}
\pagestyle{empty}

\begin{abstract}
  We consider the problem of steering a system with unknown,
  stochastic dynamics to satisfy a rich, temporally-layered task given
  as a signal temporal logic formula. We represent the system as a Markov decision process in which the states are built from a partition of the statespace and the transition probabilities are unknown. We present provably convergent reinforcement learning
  algorithms to maximize the probability of satisfying a given
  formula and to maximize the average expected robustness,
  i.e., a measure of how strongly the formula is satisfied.
  We demonstrate via a pair of robot navigation
  simulation case studies that reinforcement learning with robustness
  maximization performs better than probability maximization in terms
  of both probability of satisfaction and expected robustness.  
\end{abstract}

\section{Introduction}
We consider the problem of controlling a system with unknown,
stochastic dynamics, i.e., a ``black box'', to achieve a complex,
time-sensitive task. An example is controlling a noisy aerial vehicle
with partially known dynamics to visit a pre-specified set of regions
in some desired order while avoiding hazardous areas. We consider
tasks given as temporal logic (TL) formulae
\cite{baier2008principles}, an extension of first order Boolean logic
that can be used to reason about how the state of a system evolves
over time.  When a stochastic dynamical model is known, there exist
algorithms to find control policies for maximizing the probability of
achieving a given TL specification
\cite{Luna2014,lahijanian2015,Svorenova2015,Kamgarpour2011} by
planning over stochastic abstractions
\cite{Julius2009,Abate2011,lahijanian2015}.  However, only a handful
of papers have considered the problem of enforcing TL specifications
to a system with unknown dynamics.
Passive \cite{Brazdil2014} and active \cite{Sadigh2014RLTL,Fu2014TLRL}  
reinforcement learning has been used
to find a policy that maximizes the probability of satisfying a given
linear temporal logic formula.

In this paper, in contrast to the above works on reinforcement
learning which use propositional temporal logic, we use signal
temporal logic (STL), a rich predicate logic that can be used to
describe tasks involving bounds on physical parameters and time
intervals \cite{donze2010robust}. An example of such a property is
``Within $t_1$ seconds, a region in which $y$ is less than $\pi_1$ is
reached, and regions in which $y$ is larger than $\pi_2$ are avoided
for $t_2$ seconds." STL admits a continuous measure called
\emph{robustness degree} that quantifies how strongly a given sample
path exhibits an STL property as a real number rather than just
providing a $yes$ or $no$ answer
\cite{fainekos2009robustness,donze2010robust}. This measure enables
the use of continuous optimization methods to solve inference
(e.g., \cite{jin2013,Jones2014,kong2014temporal}) or formal synthesis
problems (e.g., \cite{Raman2014}) involving STL.
 
One of the difficulties in solving problems with TL formulae is the
history-dependence of their satisfaction. For instance, if the
specification requires visiting region A before region B, whether or
not the system should steer towards region B depends on whether or not
it has previously visited region A. For linear temporal logic (LTL) formulae with
time-abstract semantics, this history-dependence can be broken by
translating the formula to a deterministic Rabin automaton (DRA), a
model that automatically takes care of the history-dependent
``book-keeping'' \cite{ding2014,Sadigh2014RLTL}.  In the case of STL,
such a construction is difficult due to the time-bounded semantics.
We circumvent this problem by defining a fragment of STL such that the
progress towards satisfaction is checked with some finite number
$\tau$ of state measurements.  We thus define an MDP, called the
$\tau$-MDP whose states correspond to the $\tau$-step history of the
system.  The inputs to the $\tau$-MDP are a finite collection of
control actions.
  
We use a reinforcement learning strategy called $Q$-learning 
\cite{tsitsiklis1994}, in which a policy is constructed by taking actions,
observing outcomes, and reinforcing actions that improve  a given reward.  Our
algorithms  either maximize the probability of
satisfying a given STL formula, or maximize the expected robustness
with respect to the given STL formula. These procedures provably converge
 to the optimal policy for each case.  Furthermore, we
propose that maximizing expected robustness is typically more
effective than maximizing probability of satisfaction.  We prove that
in certain cases, the policy that maximizes expected robustness 
also maximizes the probability of satisfaction.  However,
if the given specification is not satisfiable, the probability
maximization will return an arbitrary policy, while the robustness
maximization will return a policy that gets as close to satisfying the
policy as possible.  Finally, we demonstrate through simulation case
studies that the policy that maximizes expected robustness in some
cases gives better performance in terms of both probability of
satisfaction and expected robustness when fewer training episodes are
available.

\section{Signal Temporal Logic(STL)}
\label{STL}

STL is defined with respect to continuously valued 
signals.  Let $\mathcal{F}(A,B)$ denote the set of mappings from $A$ to $B$  and define a \em 
signal \em  as a member of $\mathcal{F}(\mathbb{N},\mathbb{R}^n)$.  For a 
signal $s$, we denote $s^t$ as the value of $s$ at time $t$ and $s^{t_1:t_2}$ as the sequence of values $s^{t_1}s^{t_1+1}\ldots s^{t_2}$. 
Moreover, we denote $s[t]$ as the \em suffix \em from time $t$, i.e., $s[t] = \{s^{t'} | t' \geq t\}$.   

In this paper, the desired mission specification is described by an STL 
fragment with the following \em syntax \em:
\begin{equation}
\label{fragGram}
\begin{array}{rl}
\phi &:= F_{[0, T]} \psi | G_{[0, T]} \psi, \\
\psi &:= f(s) \leq d | \neg \varphi | \varphi_1 \wedge \varphi_2 | \varphi_1 U_{[a,b)} \varphi_2,
\end{array}
\end{equation}
where $T$ is a finite time bound, $\phi,\psi,$ and $\varphi$ are STL formulae,   $a$ and $b$ are non-negative real-valued constants, and  $f(s) < d$ is a predicate where $s$ is a 
signal, $f \in \mathcal{F}(\mathbb{R}^n,\mathbb{R})$ is a function, and $d \in \mathbb{R}$ is a constant.    
The Boolean operators $\neg$ and $\wedge$ are negation (``not'') and conjunction (``and''), respectively.  
The other Boolean operators are defined as usual. The temporal operators $F$, $G$, and $U$ 
stand for ``Finally (eventually)'' , ``Globally (always)'', 
and ``Until'', respectively.  Note that in this paper, we use a discrete-time version of STL rather than
the typical continuous-time formulation.

The \em semantics \em of STL is recursively defined as
\begin{equation*}
 \begin{array}{rll}
s[t] \models (f(s)< d) & \text{ iff } & f(s^t) < d \\
s[t] \models \phi_1  \wedge \phi_2  & \text{ iff } & s[t] \models \phi_1 
\text{ and } s[t] \models \phi_2 \\
s[t] \models \phi_1  \vee \phi_2  & \text{ iff } & s[t] \models \phi_1 
\text{ or } s[t] \models \phi_2 \\
s[t] \models  \LTLALWAYS_{[a,b)} \phi  & \text{ iff } & s[t'] \models 
\phi \\
&& \forall t' \in [t+a,t+b) \\
s[t] \models \LTLEVENTUALLY_{[a,b)} \phi & \text{ iff } &
\exists t'
\in [t+a,t+b)\\&& \text{s.t. } s[t'] \models 
\phi \\
s[t] \models \phi_1 U_{[a,b)} \phi_2 & \text { iff } & \exists t'
\in [t+a,t+b)\\&& \text{ s.t. } s[t''] \models \phi_1 \forall t'' \in [t,t') \\ & & 
\text{ and } s[t'] \models \phi_2.
 \end{array}
\end{equation*}
In plain English, $F_{[a,b)} \phi$ means ``within $a$ and $b$ time units in the future, $\phi$ is true,'' $G_{[a,b)} \phi$ means
``for all times between $a$ and $b$ time units in the future $\phi$ is true,'' and $\phi_1 U_{[a,b)} \phi_2$ means ``There exists a time  $c$ between $a$ and $b$ time units in the
future such that $\phi_1$ is true until $c$ and $\phi_2$ is true at $c$.''  
 STL is equipped with a \em robustness degree \em 
\cite{fainekos2009robustness,donze2010robust} (also called ``degree of 
satisfaction'') that quantifies how well a given signal $s$ satisfies a given formula $\phi$.  The robustness is calculated recursively according 
to the \em quantitative semantics \em
\begin{equation*}
 \begin{array}{rl}
 r(s,(f(s) < d),t) & = d-f(s^t) \\
 r(s,\phi_1 \wedge \phi_2,t) &= \min \big(r(s,\phi_1,t),r(s,\phi_2,t) \big) \\
 r(s,\phi_1 \vee \phi_2,t) &= \max \big( r(s,\phi_1,t),r(s,\phi_2,t) \big) \\
  r(s,\LTLALWAYS_{[a,b)} \phi ,t) & = \underset{t' \in
[t+a,t+b)}{\min} r(s,\phi ,t') \\
r(s,\LTLEVENTUALLY_{[a,b)} \phi ,t) & = \underset{t' \in
[t+a,t+b)}{\max}
r(s,\phi ,t'), \\
r(s,\phi_1 U_{[a,b)} \phi_2,t) & = \sup_{t' \in [t+a, t+b]} \Big(\min 
\big( r (\phi_2, s, t'),\\ &\inf_{t'' \in [t, t']} r(\phi_1, s, 
t'') \big) \Big).\\
 \end{array}
\end{equation*} 
We use $r(s,\phi)$ to denote $r(s,\phi,0)$. If $r(s,\phi)$ is large and 
positive, then $s$ would have to change by a large deviation in order to 
violate $\phi$.  Similarly, if $r(s,\phi)$ is large in absolute value and negative, then $s$ 
strongly violates $\phi$.

Similar to \cite{dokhanchi2014line}, let $hrz(\phi)$ denote the \emph{horizon 
length} of an STL formula $\phi$. The horizon length
is the required number of samples to resolve any (future or past) requirements 
of $\phi$. The horizon length can be computed recursively as
\begin{equation}
\label{horizonDef}
\begin{array}{rl}
hrz(p) &= 0,\\
hrz(\neg \phi) &= hrz(\phi),\\
hrz(\phi_1 \vee \phi_2) & = \max \{hrz(\phi_1), hrz(\phi_2)\},\\
hrz(\phi_1 \wedge \phi_2) & = \max \{hrz(\phi_1), hrz(\phi_2)\},\\
hrz(\phi_1 U_{[a,b]} \phi_2) & = \max \{hrz(\phi_1)+b-1, hrz(\phi_2)+b\},
\end{array}
\end{equation}
where $\phi,\phi_1,\phi_2$ are STL formulae.

\begin{example}
\label{running}
Consider the robot navigation problem illustrated in Figure
\ref{MDPEx}(a). The specification is ``Visit Regions $A$ or $B$ and
visit Regions $C$ or $D$ every 4 time units along a mission horizon of
100 units.'' Let $s(t) = \begin{bmatrix} x(t) & y(t) \end{bmatrix}^T$,
where $x$ and $y$ are the $x-$ and $y-$ components of the signal
$s$. This task can be formulated in STL as
 \begin{equation}
 \label{navSpec}
 \begin{array}{lll}
\phi =&  G_{[0,100)} & \psi \\
\psi =  & \Big( F_{[0,4)}  &  \big(
(x >2 \wedge x < 3 \wedge y >2 \wedge  y < 3) \\ 
&  &  \vee  (x >4 \wedge x < 5 \wedge y >4 \wedge y < 5) \big) \\ 
& \wedge F_{[0,4)} &   
\big( ( x >2 \wedge x < 3 \wedge y >4 \wedge  y < 5) \\  
& & \vee ( x >4 \wedge x < 5\wedge y >2 \wedge y < 3) \big) \Big). \\  
\end{array}
\end{equation}
 
Figure \ref{MDPEx}(a) shows  two trajectories of the system  beginning at the initial location of $R$ and ending in region $C$ that each 
 satisfies the inner specification $\psi$ given in \eqref{navSpec}.   Note that $s_2$
 barely satisfies $\psi$, as it only slightly penetrates region $A$,  while $s_1$ appears to satisfy it strongly, as it passes through 
 the center of region $A$ and the center of region $C$.  The
robustness degrees confirm this: $r(s_1,\psi) = 0.3$ while $r(s_2,\psi) = 0.05$.

 The horizon length of the inner specification $\psi$ of \eqref{navSpec} is 
 \begin{equation*}
  hrz(\psi) = \max \big( 4 + \max(0,0),4 + \max(0,0) \big)=4.
 \end{equation*}
\endproof
\end{example}

\section{Models for Reinforcement Learning}
\label{tmdp}
For a system with unknown and stochastic dynamics, a critical problem is how to synthesize control to achieve a desired behavior.
A typical approach is to discretize the state and action spaces of the system and then use a {\it reinforcement learning strategy},  i.e., by learning
how to take actions through trial and error interactions with an unknown environment \cite{sutton1998}.
In this section, we present models of systems that are amenable for reinforcement learning
to enforce temporal logic specifications.  We start with a discussion on the widely used LTL before introducing the particular model that we will use for reinforcement learning with STL.

\subsection{Reinforcement Learning with LTL}
One approach to the problem of enforcing LTL satisfaction in a
stochastic system is to partition the statespace and design
control primitives that can (nominally) drive the system from one
region to another.  These controllers, the stochastic dynamical
model of the system, and the quotient obtained from the partition are
used to construct a Markob decision process (MDP), called a bounded parameter MDP or BMDP,
whose transition probabilities are interval-valued \cite{Abate2011}.
These BMDPs can then be composed with a DRA constructed from a given
LTL formula to form a product BMDP.  Dynamic
programming (DP) can then be applied over this product MDP to generate a
policy that maximizes the probability
of satisfaction. Other approaches to this problem include aggregating
the states of a given quotient until an MDP can be constructed such
that the transition probability can be considered constant (with
bounded error) \cite{Lahijanian2012}.  The optimal policy can be
computed over the resulting MDP using DP \cite{Lahijanian2012b} or
approximate DP, e.g., actor-critic methods \cite{Ding2012}.

Thus, even when the stochastic dynamics of a system are known and the  logic
that encodes constraints has time-abstract semantics, the problem of constructing an abstraction
of the system that is amenable to control policy synthesis is difficult and computationally 
intensive.  Reinforcement learning methods  for enforcing
LTL constraints make the assumption that the underlying model under control is an MDP \cite{Brazdil2014,Sadigh2014RLTL,Fu2014TLRL}.   Implicitly, these procedures
compute a frequentist approximation of the transition probabilities that asymptotically approaches the true (unknown) 
value as the number of observed sample paths increases.  Since this algorithm doesn't explicitly
rely on any \em a priori \em knowledge of the transition probability, it could be applied to an
abstraction of a continuous-space system that is built from a proposition-preserving partition.
In this case, the uncertainty on the motion described by intervals in the BMDP that is reduced via computation
would instead be described by complete ignorance that is reduced via learning.  The resulting policy would
map regions of the statespace to discrete actions that will optimally drive the real-valued
state of the system to satisfy the given LTL specification.  Different partitions will result in different
policies.  In the next section, we extend the above observation to derive a discrete 
model that is amenable for reinforcement learning for STL formulae.

\subsection{Reinforcement learning with STL: $\tau$-MDP}
In order to reduce the search space of the problem, we partition the statespace of the system to form the quotient graph $\mathcal{G} = (\Sigma,E)$,
where $\Sigma$ is a set of discrete states corresponding to the regions of the statespace and $E$ corresponds to the
set of edges. An edge between two states $\sigma$ and $\sigma'$ exists in $E$ if and only if $\sigma$ and $\sigma'$ are neighbors 
(share a boundary) in the partition.
 In our case, since STL has time-bounded semantics, we cannot use an automaton with a time-abstract acceptance condition (e.g., a DRA) 
to check its satisfaction. In general, whether or not a given trajectory $s^{0:T}$ satisfies an STL formula
would be determined by directly using the qualitative semantics.  The STL fragment \eqref{fragGram} consists of a sub-formula $\psi$
with horizon length $hrz(\psi) = \tau$ that is modified by either a $F_{[0,T)}$ 
or $G_{[0,T)}$ temporal operator.  This means that in order to  update at time $t$ whether or not the given formula $\phi$ has been satisfied or violated, we can use
the $\tau$ previous state values $s^{t-\tau+1:t}$  For this reason, we choose to learn policies over an MDP with finite memory, called a $\tau$-MDP, 
whose states correspond to sequences of length $\tau$ of regions in the defined partition.

\addtocounter{example}{-1}
\begin{example}[cont'd]
Let the robot evolve according to the discrete-time Dubins dynamics
\begin{equation}
\begin{array}{c}
 x^{t+1} = x^{t} + v \delta^t \cos{\theta^t} \\
 y^{t+1} = y^{t} + v \delta^t \sin{\theta^t},
 \end{array}
\end{equation}
where $x^{t}$ and $y^{t}$ are the $x$ and $y$ coordinates of the robot at time $t$, 
$v$ is its forward speed, $\delta^t$ is a time interval,  and the robot's 
orientation is given by $\theta^t$.    The control primitives in this case are given
by $Act = \{up,down,left,right\}$ which correspond to the directions on the grid.  
Each (noisy) control primitive induces a distribution with support $\theta_{des} \pm 
\Delta\theta$, where $\theta_{des}$ is the orientation where the robot is facing the desired cell.  When a motion primitive is enacted, the robot
rotates to an angle $\theta^t$ drawn from the distribution and moves along that direction
for $\delta^t$ time units.
The partition of the statespace and the induced
quotient $\mathcal{G}$ are shown in Figures \ref{MDPEx}(b) and \ref{MDPEx}(c), respectively.  A state $\sigma_{(i,j)}$ in the quotient (Figure \ref{MDPEx}(c)) 
represents the region in the partition of the statespace (Figure \ref{MDPEx}(b)) with the point $(i,j)$ in the lower left hand corner.
\endproof
\end{example}
%
%

\begin{figure*}
\centering
 \begin{tabular}{ccc }
 \includegraphics[width=0.5\columnwidth]{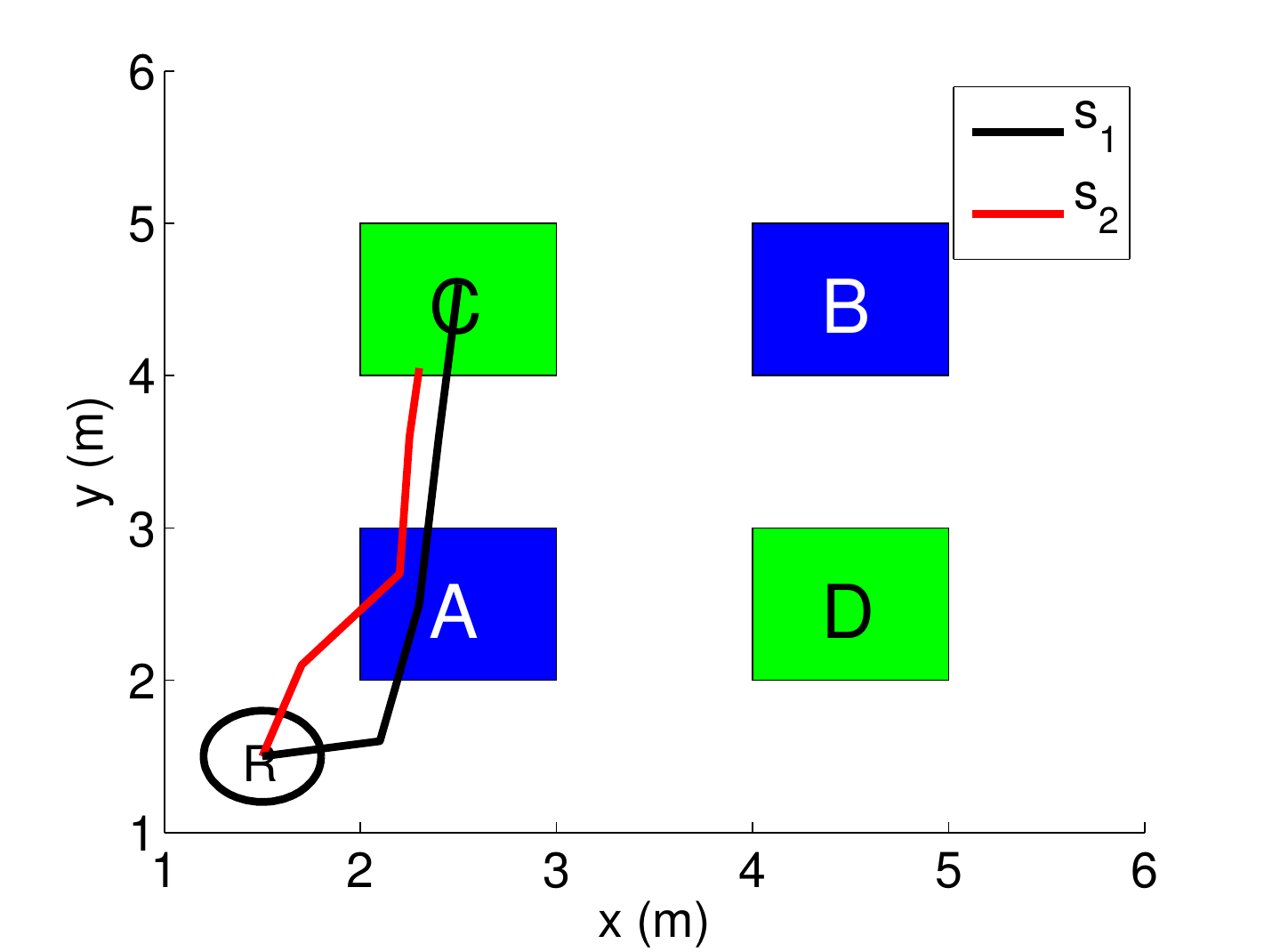} &
  \includegraphics[width=0.5\columnwidth]{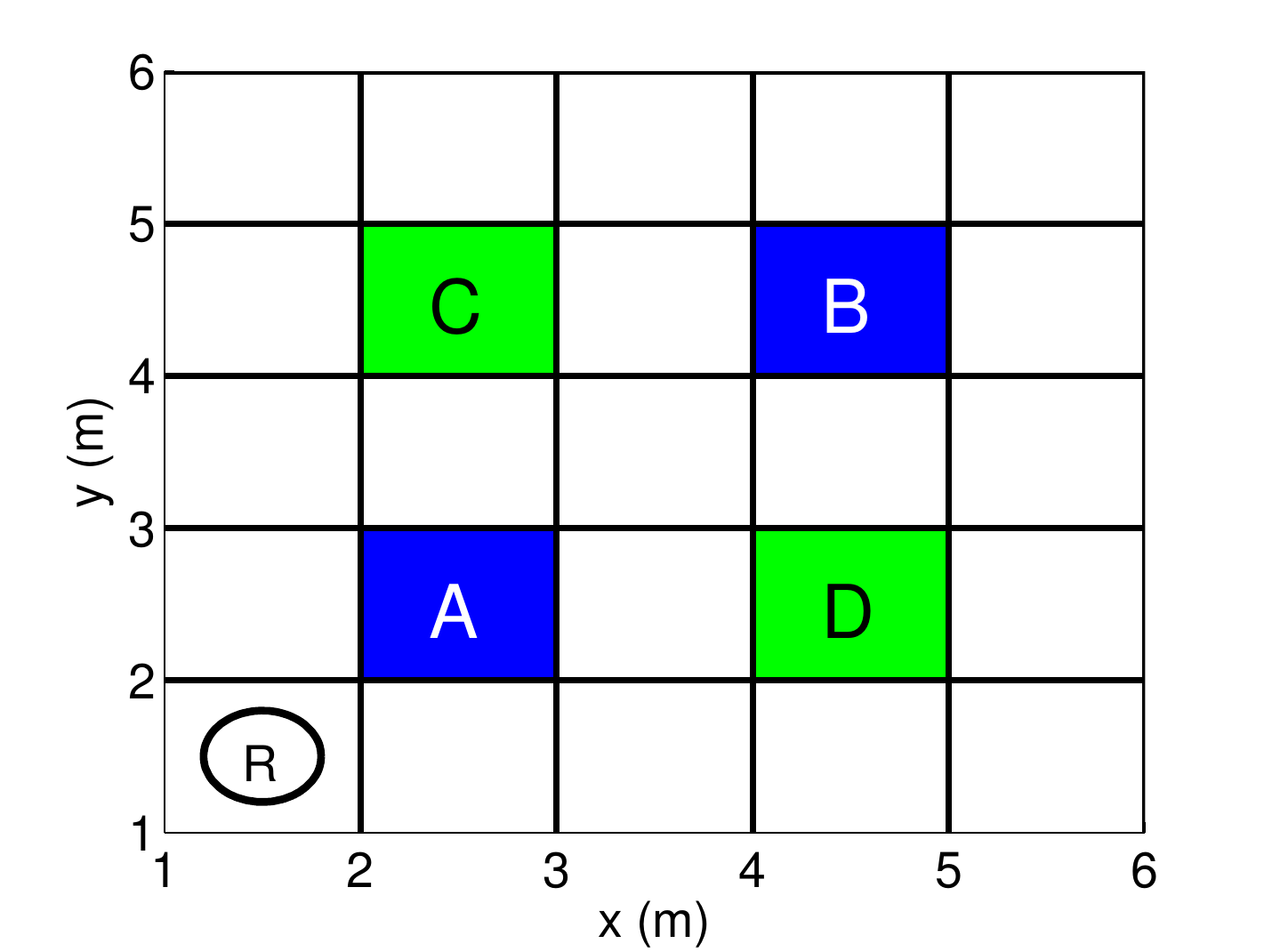} &
  \includegraphics[width=0.4\columnwidth]{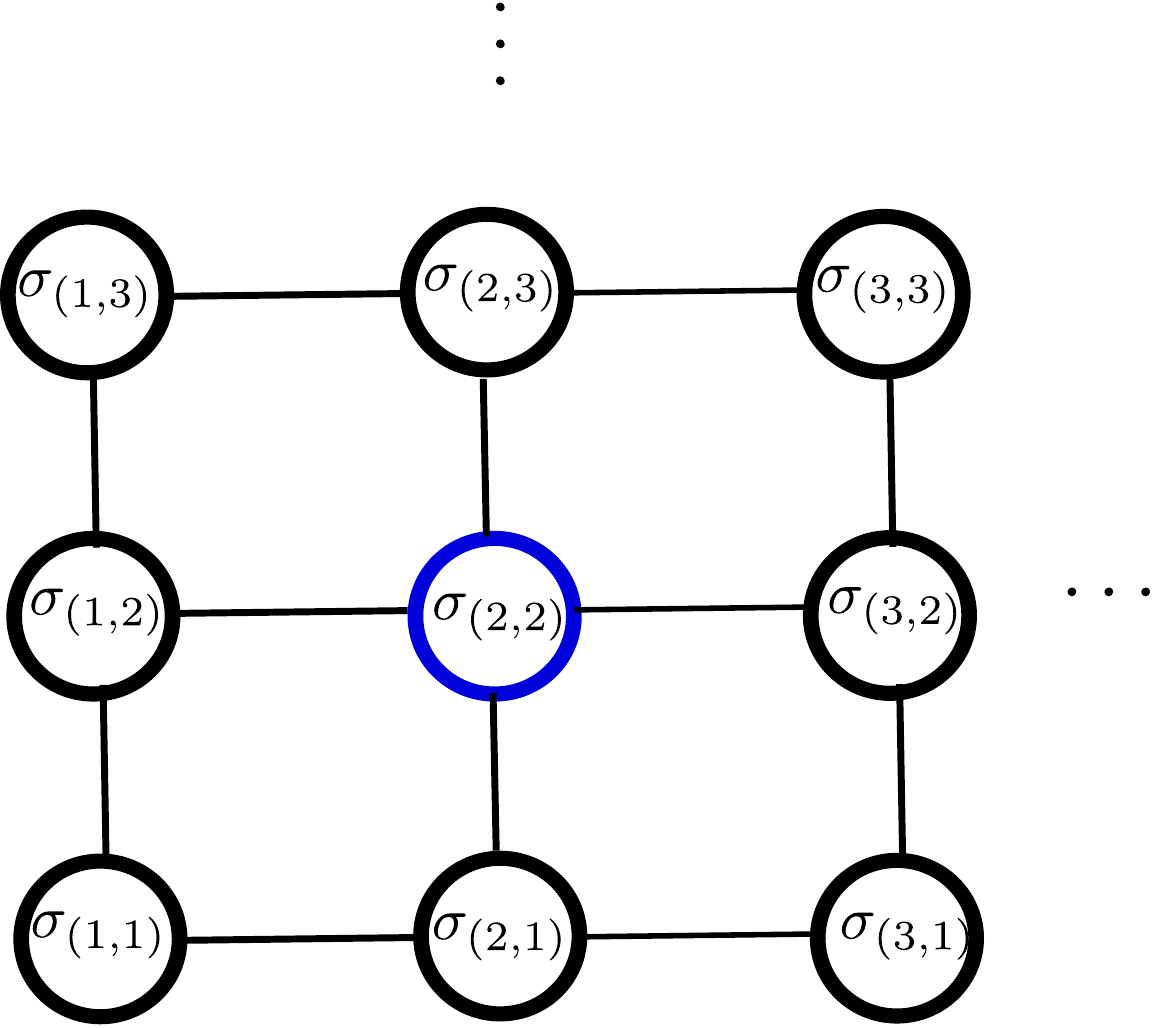} \\ 
  (a) & (b) & (c)
 \end{tabular}
\caption{\label{MDPEx}  \small (a) Example of robot navigation problem. (b) Partitioned space. (c) Subsection of the quotient.}
\end{figure*}

\begin{definition}
Given a quotient of a system $\mathcal{G}=(\Sigma,E)$ and a finite set of actions $Act$, a \emph{$\tau$-Markov Decision Process ($\tau$-MDP)} is a tuple $\mathcal{M}_{\tau}= 
\tuple{\mathcal{S},Act,\mathcal{P}}$, where 

\begin{itemize}
\item $\mathcal{S} \subseteq (\Sigma \cup \epsilon)^{\tau}$ is the set of finite states, where $\epsilon$ is the empty
string.
Each state $\sigma_{\tau} \in \mathcal{S}$ corresponds to a $\tau-$horizon (or shorter)
path in $\mathcal{G}$. Shorter paths of length $n < \tau$  (representing the case in which the system has not yet evolved for $\tau$ time steps) have
$\epsilon$ prepended $\tau-n$ times.
\item $\mathcal{P}: \mathcal{S} \times Act \times \mathcal{S} 
\rightarrow [0,1]$ is a probabilistic transition relation. $\mathcal{P}(\sigma_{\tau},a,\sigma_{\tau}')$ can be positive only 
if the first $\tau-1$ states of $\sigma_{\tau}'$ are equal to the last $\tau-1$ states of $\sigma_{\tau}$ and there exists an edge in $\mathcal{G}$
between the final state of $\sigma_{\tau}$ and the final state of $\sigma_{\tau}'$.
\end{itemize}
\label{def:tau-MDP}
\end{definition}

We denote the state of the $\tau$-MDP at time $t$ as $\sigma_{\tau}^t$. 

\begin{definition}
 Given a trajectory $s^{t-\tau+1:t}$ of the original system, we define 
its induced \em trace \em in the $\tau$-MDP $\mathcal{M}_{\tau}$ as $Tr(s^{t-\tau+1:t}) = \sigma^{t-\tau+1:t} = \sigma_{\tau}^t$. That is, $\sigma_{\tau}^t$ corresponds
to the previous $\tau$ regions of the statespace that the state has resided in from time $t-\tau+1$ to time $t$.
\end{definition}

 The construction of a $\tau$-MDP
from a given quotient and set of actions is straightforward.
The details are omitted due to length constraints.  We make the following key assumptions on the quotient and the resulting $\tau$-MDP:
\begin{itemize}
 \item The defined control
actions $Act$ will drive the system either to a point in the current region or to a point
 in a neighboring region of the partition, e.g.,no regions are ``skipped''.
 \item The transition relation $\mathcal{P}$ is Markovian.
\end{itemize}


%

For every $\tau$ state $\sigma_{\tau}^t$, there exists a continuous set of sample paths $\{s^{t-\tau+1:t}\}$ whose 
traces could be that state. The dynamics of the underlying system produces an unknown distribution
$p(s^{t-\tau+1:t}|Tr(s^{t-\tau+1:t}) = \sigma_{\tau}^t)$. Since the robustness degree is a function of sample paths of length
$\tau$ and  an STL formula $\psi$, we can define a distribution $ p(r(s^{t-\tau+1:t},\psi) | Tr(s^{t-\tau+1:t}) = \sigma_{\tau}^t)$.

\addtocounter{example}{-1}
\begin{example}[cont'd]
 Figure \ref{fig:tau_MDP} shows a portion of the $\tau$-MDP constructed from
 Figure \ref{MDPEx}.  The states in $\mathcal{M}_{4}$ are labeled
 with the corresponding sample paths of length 4 in $\mathcal{G}$. The green and blue $\sigma$'s in
 the states in $\mathcal{M}_4$ correspond to green and blue regions from Figure \ref{MDPEx}.
\end{example}

\section{Problem Formulation}
\label{problemForm}

 In this paper, we address the following two problems.
\begin{problem}[Maximizing Probability of Satisfaction]
\label{maxProb}
Let $\mathcal{M}_{\tau}$ be a $\tau$-MDP as described in the previous section. Given an STL formula $\phi$ with syntax \eqref{fragGram}, find a policy $\mu^*_{mp}\in \mathcal{F}( \mathcal{S} \times \mathbb{N},Act)$ such that
 \begin{equation}
 \begin{array}{l}
 \label{maxProbOpt}
  \mu^*_{mp} = \underset{\mu \in \mathcal{F}(\mathcal{S} \times \mathbb{N},Act)}{\arg \max}Pr_{s^{0:T}} [s^{0:T} \models \phi] \\
  
  \end{array}
 \end{equation}

\end{problem}

\begin{problem}[Maximizing Average Robustness]
\label{maxRob}
 Let $\mathcal{M}_{\tau}$ be as defined in Problem \ref{maxProb}. Given an STL formula $\phi$ with syntax \eqref{fragGram}, find a policy $\mu^*_{mr}\in \mathcal{F}( \mathcal{S},Act)$ such that
 \begin{equation}
 \label{maxRobOpt}
 \begin{array}{l}
  \mu^*_{mr} = \underset{\mu \in \mathcal{F}(\mathcal{S}\times \mathbb{N},Act)}{\arg \max}E_{s^{0:T}} [r(s^{0:T}, \phi)] \\
  \end{array}
 \end{equation}

\end{problem}

\begin{figure}[h!]
\begin{center}
\resizebox*{\columnwidth}{!}{\includegraphics{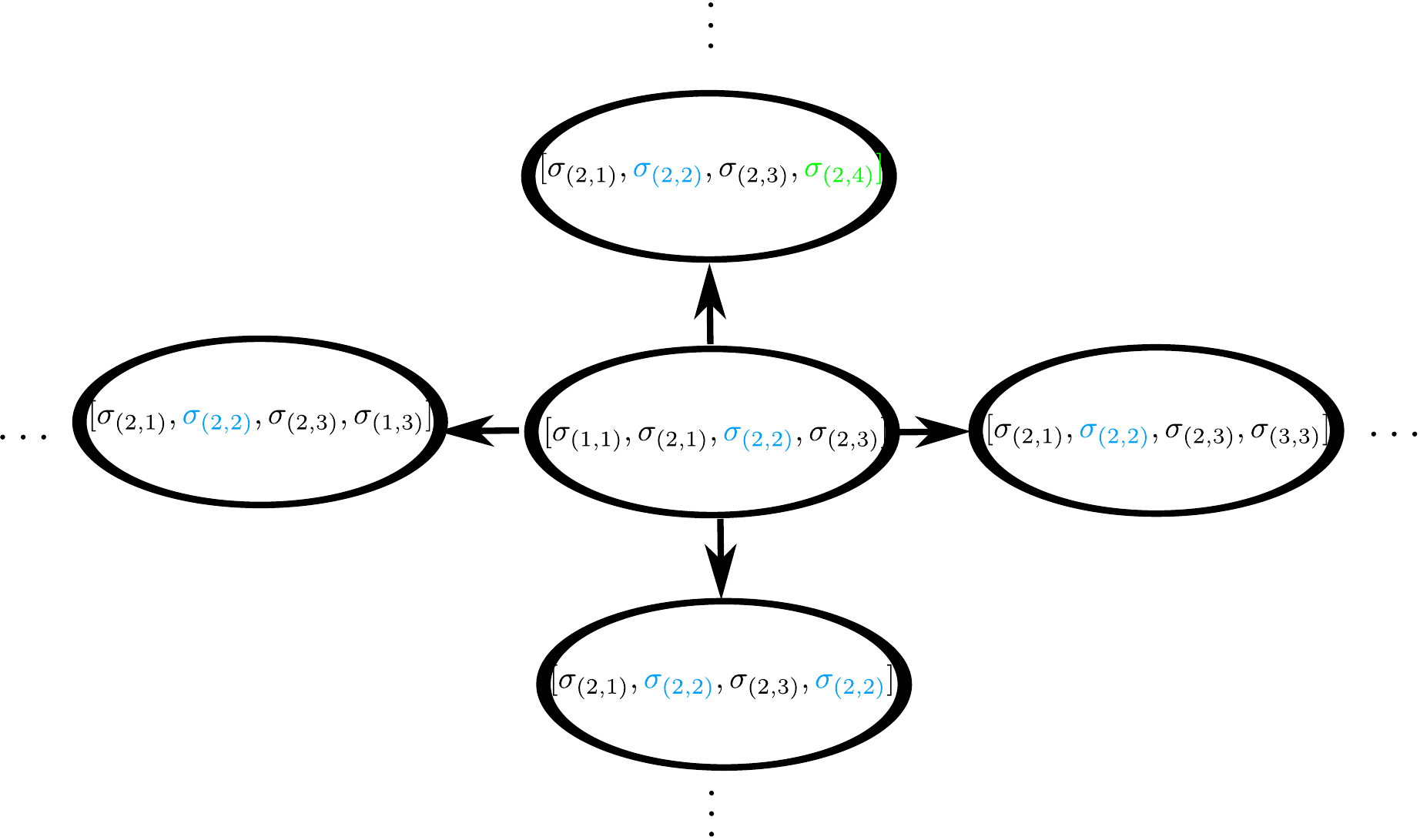}}%
\caption{\label{fig:tau_MDP} \small Part of the $\tau$-MDP constructed from the robot navigation MDP shown in Figure \ref{MDPEx}}
\end{center}
\end{figure}

%

Problems \ref{maxProb} and \ref{maxRob} are two alternate solutions to enforce a given STL specification.  The policy found by Problem~\ref{maxProb}, i.e. $\mu^*_{mp}$, maximizes the chance that $\phi$ will be satisfied, while the policy found by Problem~\ref{maxRob}, i.e. $\mu^*_{mr}$, drives the system to satisfy $\phi$ as strongly as possible on average. Problems similar to \eqref{maxProbOpt} have already been considered in the literature (e.g., \cite{Fu2014TLRL,Sadigh2014RLTL}). However, Problem \ref{maxRob} is a novel formulation that provides some advantages over Problem \ref{maxProb}. 
As we show in Section \ref{equivRelation}, for some special systems, $\mu^*_{mr}$ achieves the same probability of satisfaction as $\mu^*_{mp}$.  Furthermore, if $\phi$ is not satisfiable, any arbitrary policy could be a solution to Problem \ref{maxProb}, as all policies will result in a satisfaction probability of 0.  If $\phi$ is unsatisfiable, Problem \ref{maxRob} yields a solution that attempts to get as close as possible to satisfying the formula, as the optimal solution will have an average robustness value that is least negative.

The forms of the objective functions differ for the  two different types of formula,
$\phi = F_{[0,T)} \psi$ and $\phi=G_{[0,T)} \psi$. 

\noindent\emph{Case 1:} \label{case1}Consider an STL formula $\phi = F_{[0,T)} \psi$.  In this case, the 
objective function in \eqref{maxProbOpt} can be rewritten as
\begin{equation}
\label{eq:Pr_case1}
Pr_{s^{0:T}} [ \exists t = \tau,\ldots,T-\tau \text{ s.t. } s^{t-{\tau+1}:t} \models \psi], 
\end{equation}
and the objective function in \eqref{maxRobOpt} can be rewritten as
\begin{equation}
\label{eq:Rob_case1}
 E_{s^{0:T}} [\underset{t=\tau,\ldots,T-\tau}{\max} r(s^{t-\tau+1:t}, \psi)].
\end{equation}


\noindent\emph{Case 2:} Now, consider an STL formula $\phi=G_{[0,T)}\psi$. The 
objective function in \eqref{maxProbOpt} can be rewritten as
\begin{equation}
Pr_{s^{0:T}} [ \forall t = \tau,\ldots,T-\tau \text{, } s^{t-{\tau}:t} \models \psi], 
\end{equation}
and the objective function in \eqref{maxRobOpt} can be rewritten as
\begin{equation}
 E_{s^{0:T}} [\underset{t=\tau,\ldots,T-\tau}{\min} r(s^{t-\tau+1:t}, \psi)].
\end{equation}

\section{Maximizing Expected Robustness vs. Maximizing Probability of Satisfaction}
\label{equivRelation}

Here, we demonstrate that the solution to \eqref{maxRobOpt} subsumes the solution to \eqref{maxProbOpt} for a certain class of systems.     Due to space 
limitations, we only consider formulae of the type $\phi = F_{[0,t)}\psi$.
Let $\Mt=(\St,\mathcal{P}_{\tau},Act)$ be a $\tau$-MDP.    For simplicity, we make the following
assumption on $\St$.
\begin{assump}
\label{stateSpaceAssumption}
For every state $\st \in \St$,  either every trajectory $s^{t+\tau-1:t}$
whose trace is $\st$ satisfies $\pt$,  
denoted $\st \models \pt$, or  every trajectory that passes through the sequence of regions associated with $\st$ does not satisfy  $\pt$, 
denoted $\st \not \models \pt$.
\end{assump}
Assumption \ref{stateSpaceAssumption} can be enforced in practice during partitioning.  
 We define the set 
 \begin{equation}
A = \{\st \in \St| \st \models \pt\}.
\end{equation}

\begin{definition}\label{def:distance}
The  signed graph distance of a $\tau-$state $\sigma_{\tau}^i \in \mathcal{S}$ to a set $X \subseteq \mathcal{S}$ is
\begin{equation}
d(\sigma_{\tau}^i,X)= \left \{ \begin{array}{ll} \min\limits_{\sigma_{\tau}^j \in X} l(\sigma_{\tau}^i,\sigma_{\tau}^j)  & \st^i \not \in X \\
                              - \min\limits_{\sigma_{\tau}^j \in  \St \setminus X} l(\st^j,\st^i)  & \st^i \in X \\
                               \end{array} \right .
\end{equation}
where $l(\sigma_{\tau}^i,\sigma_{\tau}^j)$ is the length of the shortest path from $\sigma_{\tau}^i$ to $\sigma_{\tau}^j$.
\end{definition}

We also make the following two assumptions.
\begin{assump}
\label{assump:boundedRob}
 For any signal $s^{t-\tau+1:t}$ such that $Tr(s^{t-\tau+1:t}) \in \St$, let 
 $r(s^{t-\tau+1:t},\pt)$ be bounded from below by $R_{min}$ and from above by 
 $R_{max}$.
 \end{assump}

\begin{assump}
 Let $D_{\st}(\delta) = Pr[r(s^{t:t+\tau},\pt) > \delta | Tr(\sigt) = \st]$.  For any two states, 
\begin{equation}
d(\sigma^i_{\tau},A) < d(\st^j,A) \Rightarrow D_{\st^i}(\delta) \geq D_{\st^j}(\delta) \text{  } \forall \delta \in [R_{min},R_{max}]
\end{equation}
\label{assump:dr}
\end{assump}

  Now we define the policies $\mu^*_{mp}$ and $\mu^*_{mr}$ over $\Mt$ as 
 \begin{equation}
 \begin{array}{c}
 \label{maxProbOpt2}
  \mu^*_{mp} = \underset{\mu \in \mathcal{F}(\mathcal{S} \times \mathbb{N},Act)}{\arg \max}Pr_{\sigma_{\tau}^{0:T}} \Big[ \exists t \in [0,T] \text{ s.t. } \sigma_{\tau}^{t} \models \psi\Big]  \\ 
  \end{array}
 \end{equation}
  \begin{equation}
 \label{maxRobOpt2}
 \begin{array}{c}
  \mu^*_{mr} = \underset{\mu \in \mathcal{F}(\mathcal{S}\times \mathbb{N},Act)}{\arg \max} E_{\sigma_{\tau}^{0:T}} \Big[\underset{t=0,\ldots,T}{\max} r(\sigma_{\tau}^{t}, \psi)\Big]  \\ 
\end{array}
 \end{equation}

\begin{proposition}
If Assumptions~\ref{stateSpaceAssumption},\ref{assump:boundedRob}, and \ref{assump:dr} hold, then the policy $\mu^*_{mr}$ maximizes the expected probability of satisfaction.
\label{lemma:maxRob_minDistA}
\end{proposition}
\begin{proof}
Given any policy $\mu$, its associated reachability probability can be defined as 
\begin{equation}
 Pr_{\mu}(\st) = Pr_{\mu}\left[\st  = \underset{\st^0,\ldots,\st^{T-\tau}}{\arg \min} d(\st,A)\right].
\end{equation}
Let $I(.)$ be the indicator function such that $I(B)$ is 1 if $B$ is true and $0$ if $B$ is false. By definition, the expected probability of satisfaction for a given policy $\mu$ is
\begin{equation}
\begin{array}{ll}
 EPS(\mu) & = E\big[ I(\exists 0<k<T-\tau \text{ s.t. } \st^k \models \pt) \big] \\
  & =  \sum\limits_{\st \in \St} Pr_{\mu}(\st)I(\st \in A)  \\
  & = \sum\limits_{\st \in A} Pr_{\mu}(\st).
  \end{array}
 \end{equation}
Also, the expected robustness of policy $\mu$ becomes
\begin{equation}
\label{eq:ER_mu}
\begin{array}{lcl}
 ER(\mu) & = &  E\big[\underset{k=0,\ldots,T-\tau}{\max} r(\st^k,\pt) \big] \\
 & = &  \int_{0}^{R_{max}} Pr\big[ \underset{k=0,\ldots,T-\tau}{\max} r(\st^k,\pt) > x \big]dx  + \\
  & & \int_{R_{min}}^0 1 - Pr\big[ \underset{k=0,\ldots T-\tau}{\max} r(\st^k,\pt) > x \big]dx \\
 & = & \int_{0}^{R_{max}} Pr\big[ \underset{k=0,\ldots,T-\tau}{\max} r(\st^k,\pt) > x \big]dx - \\
 & & \int_{R_{min}}^{0} Pr\big[ \underset{k=0,\ldots,T-\tau}{\max} r(\st^k,\pt) > x \big]dx - R_{min} \\
  & = & \int_{0}^{R_{max}} \sum\limits_{\st \in \St} Pr_{\mu}(\st) D_{\st}(x) dx - \\
 &  & \int_{R_{min}}^{0} \sum\limits_{\st \in \St} Pr_{\mu}(\st) D_{\st}(x) dx - R_{min} \\
 & = & \sum\limits_{\st \in A} Pr_{\mu}(\st) \int_{0}^{R_{max}} D_{\st}(x)dx - \\
   & & \sum\limits_{\st \not \in A} Pr_{\mu}(\st) \int_{R_{min}}^0 D_{\st}(x)dx - R_{min}.
 \end{array}
\end{equation}

Since $R_{min}$ is constant, maximizing \eqref{eq:ER_mu} is equivalent to 
\begin{equation}
\label{eq:ER_mu_2}
 \begin{array}{c}
\underset{\mu}{\max} \Big(  \sum\limits_{\st \in A} Pr_{\mu}(\st) \int_{0}^{R_{max}} D_{\st}(x)dx \\
  -\sum\limits_{\st \not \in A} Pr_{\mu}(\st) \int_{R_{min}}^0 D_{\st}(x)dx \Big) 
 \end{array}
\end{equation}

Let $p$ be the satisfaction probability such that $p = \sum\limits_{\st \in A} Pr_{\mu}(\st)$. Then, we can rewrite the objective in \eqref{eq:ER_mu_2} as

\begin{equation}
\begin{array}{lcl}
 J(\mu) & = &  p\sum\limits_{\st \in A} Pr_{\mu} \big[\st = \underset{\st^0,\ldots,\st^{T-\tau}}{\arg \min} d(\st,A)| \st \in A \big] \\  & & \times \int_{0}^{R_{max}}D_{\st}(x)dx\\
 &  & - (1-p) Pr_{\mu} \big[ \st = \underset{\st^0,\ldots,\st^{T-\tau}}{\arg \min} d(\st,A)| \st \not \in A \big] \\  & & \times\int_{R_{min}}^0D_{\st}(x)dx . 
\end{array}
\end{equation}

Now,
\begin{equation}
\label{eq:Jp}
\begin{array}{lcl}
 \frac{\partial J(\mu)}{\partial p}   & = &  \sum\limits_{\st \in A} Pr_{\mu} \big[ \st = \underset{\st^0,\ldots,\st^{T-\tau}}{\arg \min} d(\st,A)| \st \in A \big] \\  & & \times \int_{0}^{R_{max}}D_{\st}(x)dx\\
 &  &  + Pr_{\mu} \big[ \st = \underset{\st^0,\ldots,\st^{T-\tau}}{\arg \min} d(\st,A)| \st \not \in A \big] \\  & & \times\int_{R_{min}}^0D_{\st}(x)dx \\
 & > & 0 
\end{array}
 \end{equation} 
 Thus, any policy $\mu$ increasing $J(\mu)$ also leads to an increase in $p$. Since increasing $J(\mu)$ is equivalent to increasing $ER(\mu)$, then we can conclude that the policy that maximizes the robustness also achieves the  maximum satisfaction probability.
\end{proof}

\section{Control Synthesis to Maximize Robustness}
\label{controlSyn}

\subsection{Policy Generation through $Q$-Learning}

Since we do not know the dynamics of the system under control, we cannot \em a priori \em predict how a given control action will affect the evolution of the 
system and hence its progress towards satisfying/dissatisfying a given specification. Thus, we use the well-known paradigm of \em reinforcement learning \em
to learn policies to solve Problems \ref{maxProb} and \ref{maxRob}.  In reinforcement learning, the system takes actions and
records the rewards associated with the state-action pair.  These rewards are then used to update a feedback policy that maximizes the 
expected gathered reward.  In our cases, the rewards that we collect over $\mathcal{M}_{\tau}$ are related to whether or not $\psi$
is satisfied (Problem \ref{maxProb}) or how robustly $\psi$ is satisfied/violated (Problem \ref{maxRob}). 

Our solutions to these problems rely on a $Q$-learning formulation \cite{tsitsiklis1994}.  Let $R(\st^t,a)$ be the reward collected when
action $a \in Act$ was taken in state $\st^t\in \mathcal{S}$.  Define the function $Q: \mathcal{S} \times Act \times \mathbb{N}$
as 
\begin{equation}
\begin{array}{lcl}
 Q(\st^{T-t},a,t) & = &  R(\st^{T-t},a) +  \\ & &  \underset{\{\mu_l \in \}_{l=T-t-1}^T}{\max}  E \big[ \sum_{l=T-t-1}^{T} R(\st^{l},\mu_l(\st^{l})) \big]\\
  & = & R(\st^{T-t},a) + \underset{a' \in Act}{\max}Q(\st^{T-t+1},a',t-1).
 \end{array}
\end{equation}

For an optimization problem with a cumulative objective function of the form 
\begin{equation}
 \label{additiveObjective}
 \sum_{l=\tau:T} R(\st^{l},a^l),
\end{equation}
the optimal policy $\mu^* \in \mathcal{F}(\mathcal{S},Act)$ can be found by
\begin{equation}
 \mu^*(\st^t,T-t) = \underset{a \in Act}{\arg \max} Q(\st^t,a,T-t).
\end{equation}

Applying the update rule
\begin{equation}
\label{Qlearn}
\begin{array}{ll}
Q_{t+1}(\st^t,a^t,T-t) = & (1-\alpha_t) Q_t(\st^t,a^t,T-t) +  \\ & \alpha_t [R(\st^t,a^t)+  \gamma \max\limits_{a' \in A} Q_t (\st^{t+1},a')]
\end{array}
\end{equation}
 where $0 < \gamma < 1$ will cause $Q_{t}$ converges to $Q$ w.p. 1 as $t$ goes to infinity \cite{tsitsiklis1994}.
%
%
%

\subsection{ Batch $Q$-learning}
\label{batchQ}
We cannot reformulate Problems \ref{maxProb} and \ref{maxRob} into the form \eqref{additiveObjective} 
(see Section \ref{problemForm}). Thus, we 
propose an alternate $Q-$learning formulation, called {\it batch $Q$-learning }, to solve these problems.  Instead
of updating the $Q$-function after each action is taken, we wait until an entire episode $s^{[0:T)}$ is completed before 
updating the $Q$-function.  The batch $Q$-learning procedure is summarized in Algorithm \ref{batchQAlg}.

\begin{algorithm}
 \caption{\label{batchQAlg} The Batch $Q$ learning algorithm.}
 \begin{algorithmic}
 \FUNCTION{BatchQLearn}{$Sys$,probType,$N_{ep}$,$\phi$}
  \STATE $Q \leftarrow $ RandomInitialization \\
  \STATE $\mu \leftarrow$ InitializePolicy($Q$)
  \FOR{$n = 1$ to $N_{ep}$} 
  \STATE $s^{[0,T)} \leftarrow $Simulate$(Sys,\mu)$ \\
  \STATE $Q \leftarrow $ UpdateQFunction($Q$,$\mu$,$s^{0:T}$,$\phi$,probType) \\
  \STATE $ \mu \leftarrow $ UpdatePolicy($\mu,Q$) 
  \ENDFOR
  \RETURN{$Q$,$\mu$}
 \end{algorithmic}

\end{algorithm}

\begin{algorithm}
 \caption{\label{QUpdate} Function used to update $Q$ function used in Algorithm \ref{batchQAlg}.}
 \begin{algorithmic}
  \FUNCTION{UpdateQFunction}{$Q$,$\mu$,$s^{0:T}$,$\phi$,$\gamma$,probType}
  \FOR {$n = T-\tau-1$ to $\tau$}
  \IF{probType is MaximumProbability} 
  \STATE $Q_{tmp}(\st^n,\mu(\st^n,T-n)) \leftarrow $ \\
  \STATE $\max(I(s^{n-\tau+1:n} \models \phi)$, \\ 
  \STATE $\gamma Q_{tmp}(\st^{n+1},\mu(\st^{n+1},T-n-1))$
  \ELSE
  \STATE $Q_{tmp}(\st^n,\mu(\st^n,T-n)) \leftarrow $ \\
  \STATE $\max(r(s^{n-\tau+1:n}, \phi)$, \\ 
  \STATE $\gamma Q_{tmp}(\st^{n+1},\mu(\st^{n+1},T-n-1))$
  \ENDIF
  \STATE $Q_{new}(\st^n,\mu(\st^n,T-n) \leftarrow $
  \STATE $(1-\alpha)Q_{tmp}(\st^n,\mu(\st^n,T-n)$ \\
  \STATE $ + \alpha Q(\st^n,\mu(\st^n,T-n)$
  \ENDFOR
  \RETURN{$Q_{new}$}
 \end{algorithmic}
\end{algorithm}

The $Q$ function is initialized to random values and $\mu$ is computed from the initial $Q$ 
values. Then, for $N_{ep}$ episodes, the system is simulated using $\mu$. Randomization
is used to encourage exploration of the policy space.  The observed trajectory
is then used to update the $Q$ function according to Algorithm \ref{QUpdate}.  The new value of the $Q$ function is used
to update the policy $\mu$.   For compactness, Algorithm \ref{QUpdate} as written only
covers the case $\phi = F_{[0,T)} \psi$.  The case in which $\phi=G_{[0,T)} \psi$ can be addressed similarly.

\subsection{Convergence of Batch $Q$-learning}
Given a formula of the form $\phi= F_{[0,T)} \psi$ and an objective of maximizing the expected robustness (Problem \ref{maxRob}), we will show that applying Algorithm \ref{batchQAlg} converges to the optimal solution. The
other three cases discussed in Section~\ref{problemForm} can be proven similarly.   The following analysis is based on \cite{Melo2015}.
The optimal $Q$ function derived from \eqref{eq:Rob_case1} is 
\begin{equation}
\label{optQ}
\begin{array}{ll}
 Q^*(\st^k,a,T-k) = &  \sum_{\st^{t+1}}\mathcal{P}(\st^t,a,\sigma^{t+1}) \max(r(\st^t,\pt), \\
  & \underset{b \in Act}{\max}\gamma Q^*(\st^{t+1},b,T-t-1)).
  \end{array}
\end{equation}

This gives the following convergence result.

\begin{proposition}
\label{Qlearnconverge}
The $Q$-learning rule given by 
\begin{equation}
\label{qrule}
\begin{array}{ll}
 Q_{k+1}(\st^t,a^t,T-t) = &  (1-\alpha_k) Q_k(\st^t,a^t,T-t) \\
 & + \alpha_k \max(r(\st^t,\pt), \\ & \underset{b \in Act}{max}  \gamma Q_k(\st^{t+1},b,T-t-1)),
 \end{array}
\end{equation}
converges to the optimal $Q$ function \eqref{optQ} if the sequence $\{\alpha_k\}_{k=0}^{\infty}$ is such that 
 $\sum_{k=0}^{\infty} \alpha_k = \infty$ and  $\sum_{k=0}^{\infty}(\alpha_k)^{2} < \infty$.
\end{proposition}

\begin{proof}
 (Sketch) The proof of Proposition \ref{Qlearnconverge} relies primarily on Proposition \ref{contProp}.  Once this 
is established, the rest of the proof varies only slightly from the presentation in \cite{Melo2015}.
\end{proof}

Note that in this case, $k$ ranges over the number of episodes and $t$ ranges over the time 
coordinate of the signal.
\begin{proposition}
\label{contProp}
The optimal $Q$-function given by \eqref{optQ} is a fixed point of the contraction mapping $H$ where
\begin{equation}
\begin{array}{ll}
 (Hq)(\st^t,a,T-t) =  & \sum_{\st^{t+1}}\mathcal{P}(\st^t,a,\sigma^{t+1})\max(r(\st^t,\pt),  \\
 & \gamma \underset{b \in Act}{\max}q(\st^{t+1},b,T-t-1)).
 \end{array}
\end{equation}
\end{proposition}
\begin{proof}
By \eqref{optQ}, if $H$ is a contraction mapping, then $Q^*$ is a fixed point of $H$.  Consider
 \begin{equation}
 \label{contMap}
  \begin{array}{ll}
   ||Hq_1-Hq_2||_{\infty} & = \max\limits_{\st,a} \sum\limits_{\st'} \mathcal{P}(\st,a,\st')(\max(r(\st,\pt), \\
    &  \gamma  \underset{b \in Act}{\max} q_1(\st',b,T-t-1)) \\ & -\max(r(\st,\pt),\underset{b \in Act} \gamma q_2(\st',b,T-t-1)).
  \end{array}
 \end{equation}

Define  
 \begin{equation}
   q_j^*(t) = \underset{b \in Act}{\max} \gamma q_1(\st',b,t).
  \end{equation}
  WOLOG let $ q_1^*(T-t-1) >  q_2^*(T-t-1)$. Define
  \begin{equation}
  \begin{array}{ll}
   R(\st') = &  (\max(r(\st,\pt),q_1^*(T-t-1) \\ 
    & -\max(r(\st,\pt),q_2^*(T-t-1))
   \end{array}
  \end{equation}
  
There exist 3 possibilities for the value of $R(\st)$. 
 \begin{subequations}
  \begin{equation}
    \begin{array}{l} r(\st,\pt) > q_1^*(T-t-1) > q_2^*(T-t-1) \\
    \Rightarrow R(\st')  = 0 \end{array}.
  \end{equation}
\begin{equation}
  \begin{array}{l}  q_1^*(T-t-1) > r(\st,\pt)  >  q_2^*(T-t-1) \\
    \Rightarrow R(\st') = || q_1^*(T-t-1) -r||_{\infty} <  \gamma ||q_1-q_2||_{\infty} \end{array}
\end{equation}
\begin{equation}
   \begin{array}{l}  q_1^*(T-t-1) >  q_2^*(T-t-1) > r(\st,\pt) \\
   \Rightarrow R(\st') <  \gamma ||q_1-q_2||_{\infty}.\end{array}
\end{equation}
 \end{subequations}

Thus, this means that $R(\st') \leq \gamma ||q_1-q_2||_{\infty}$  $\forall \st'$.  Hence,
\begin{equation} 
\begin{array}{ll} 
  ||Hq_1-Hq_2||_{\infty}  &  = \max_{\st,a} \sum_{\st'} \mathcal{P}(\st,a,\st')R(\st')  \\
   & \leq \max_{\st,a}\sum_{\st'}\mathcal{P}(\st,a,\st') \gamma ||q_1-q_2||_{\infty}  \\ 
    & \leq  \gamma ||q_1-q_2||_{\infty}.
    \end{array}
\end{equation}
Therefore, $H$ is a contraction mapping.
\end{proof}

\section{Case Study}
\label{caseStudy}

We implemented the batch-$Q$ learning algorithm (Algorithm
\ref{batchQAlg}) and applied it to two case studies that adapt the
robot navigation model from Example \ref{running}.  For each case
study, we solved Problems \ref{maxProb} and \ref{maxRob} and compared
the performance of the resulting policies.  All simulations were
implemented in Matlab and performed on a PC with a 2.6 GHz processor
and 7.8 GB RAM.

\subsection{Case Study 1: Reachability}
\label{reachability}
First, we consider a simple reachability problem.  The  given STL specification is
\begin{equation}
 \label{casestudy1spec}
 \phi_{cs1} = F_{[0,20)}(F_{[0,1)} \varphi_{blue} \wedge G_{[1,4)}  \neg \varphi_{blue}), 
\end{equation}
where $\varphi_{blue}$ is the STL subformula corresponding to being in a blue region.
In plain English, \eqref{casestudy1spec} can be stated
as ``Within 20 time units, reach a blue region and then don't revisit a blue region
for 4 time units.''  The results from applying
Algorithm \ref{batchQAlg} are summarized in Figure \ref{caseStudy1}.  We used the
parameters $\gamma = 1, \alpha_t = 0.95$, $N_{ep}$ = 300 and
$\epsilon^t = 0.995^t$, where $\epsilon^t$ is the probability at
iteration $t$ of selecting an action at random \footnote{Although the conditions $\gamma < 1$ and $\sum_{k=0}^{\infty}\alpha_k^2 < \infty$ are technically
required to prove convergence, in practice these conditions can be relaxed without having adverse effects 
on learning performance}.  Constructing the
$\tau$-MDP took 17.2s.  Algorithm \ref{batchQAlg} took 161s to solve
Problem \ref{maxProb} and 184s to solve Problem \ref{maxRob}.


\begin{figure}
\begin{tabular}{cc}
\includegraphics[width=0.43\columnwidth]{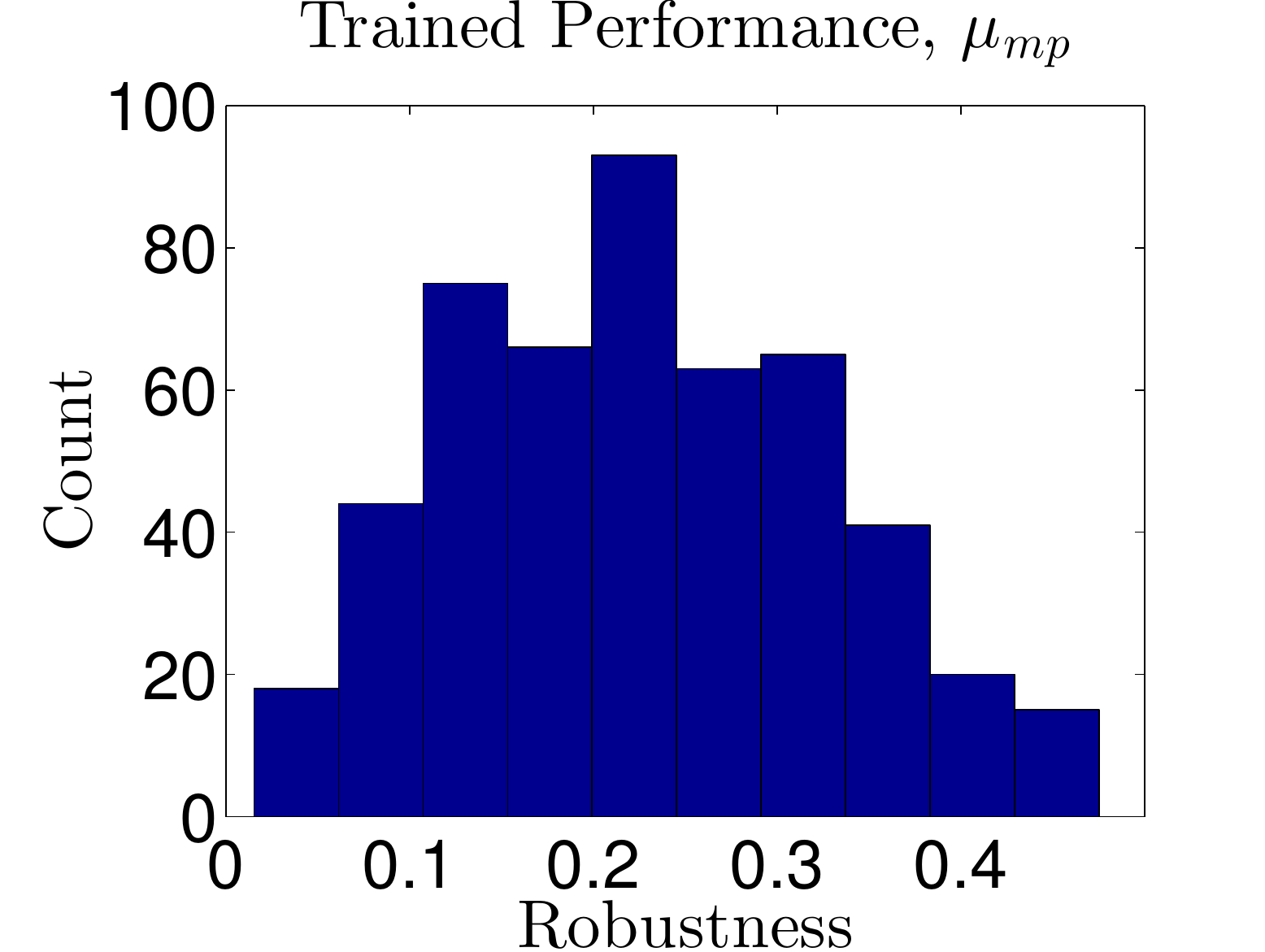} & \includegraphics[width=0.43\columnwidth]{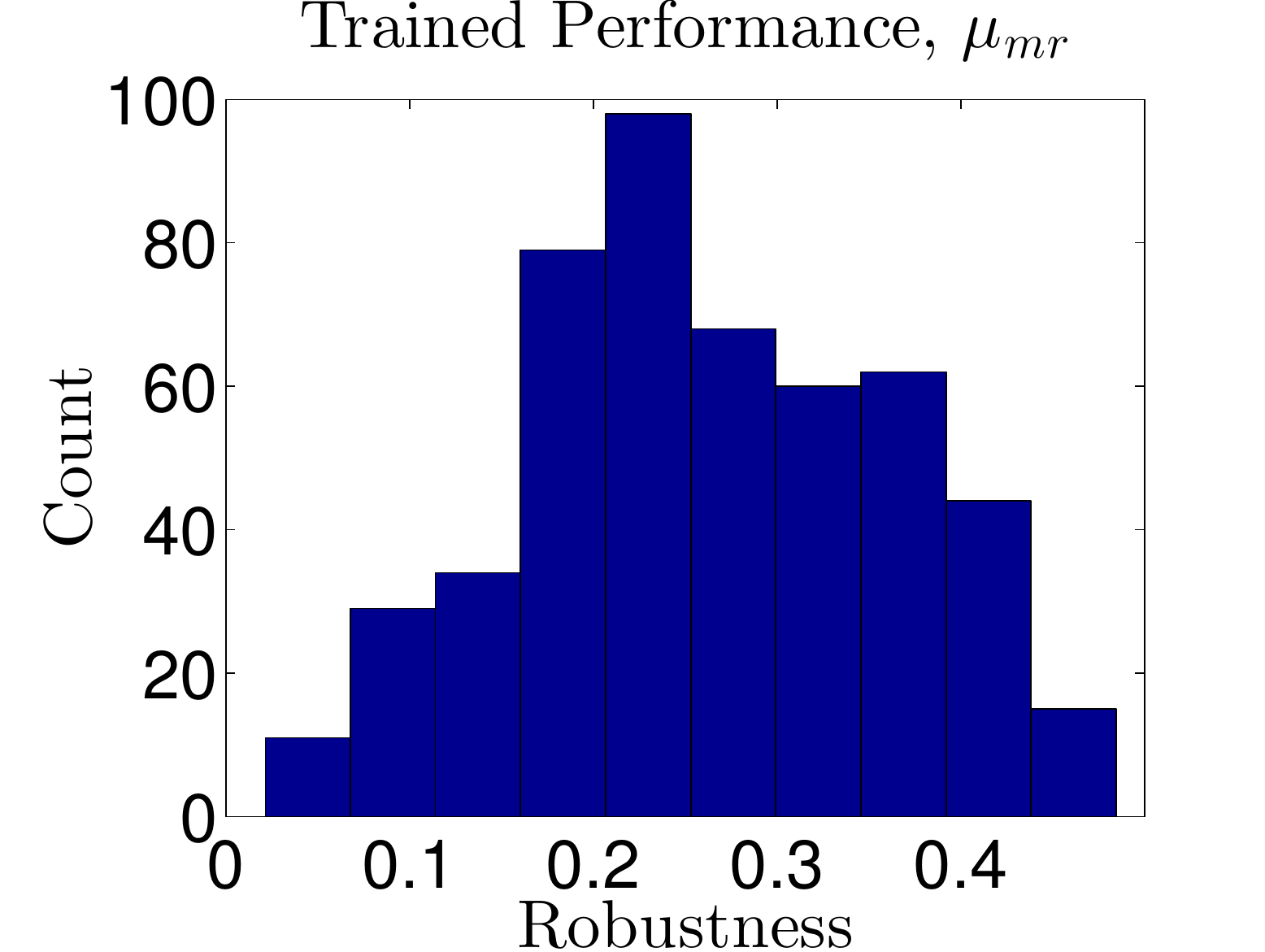} \\
 (a) & (b) \\
  \includegraphics[width=0.43\columnwidth]{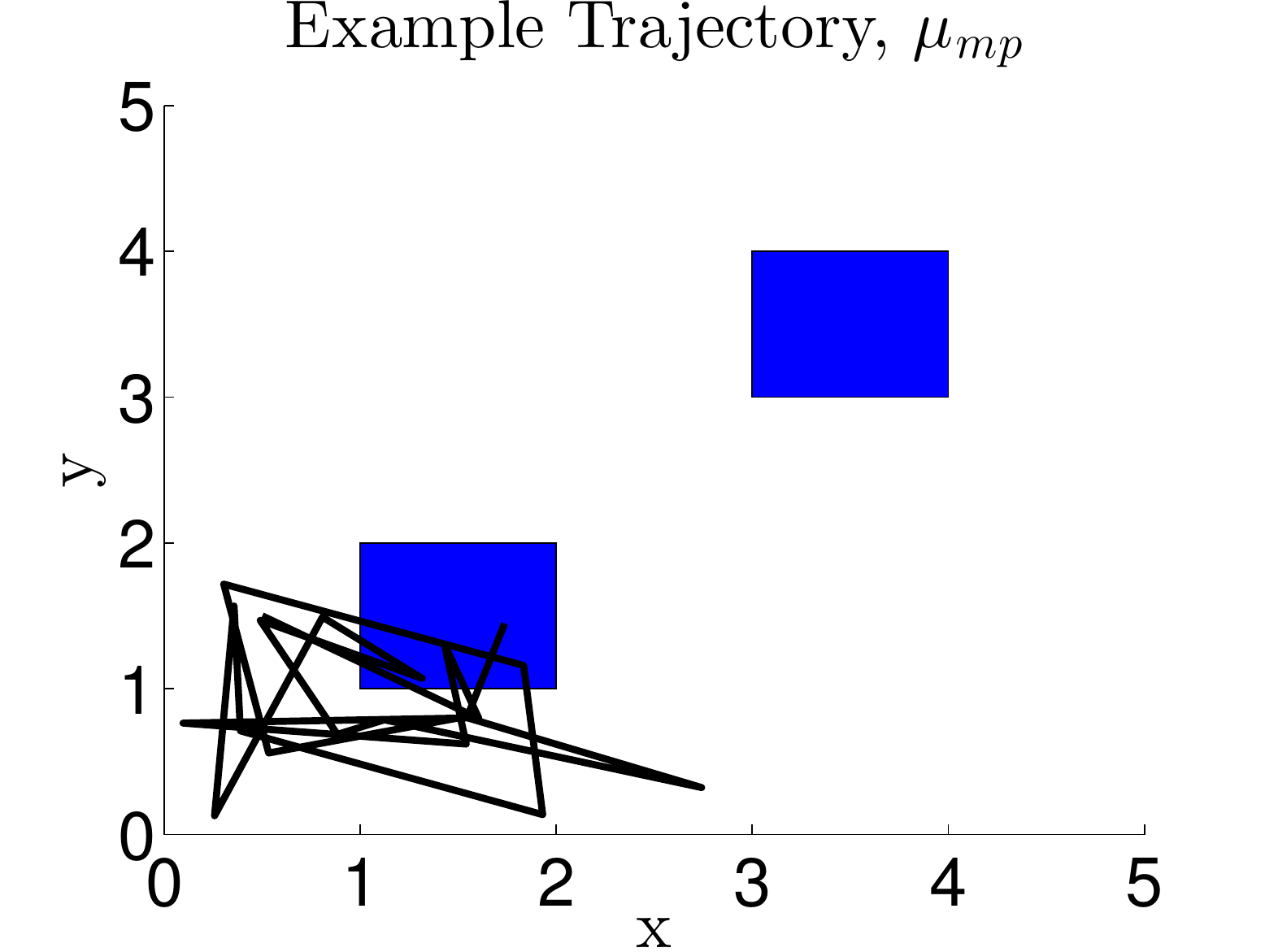} & \includegraphics[width=0.43\columnwidth]{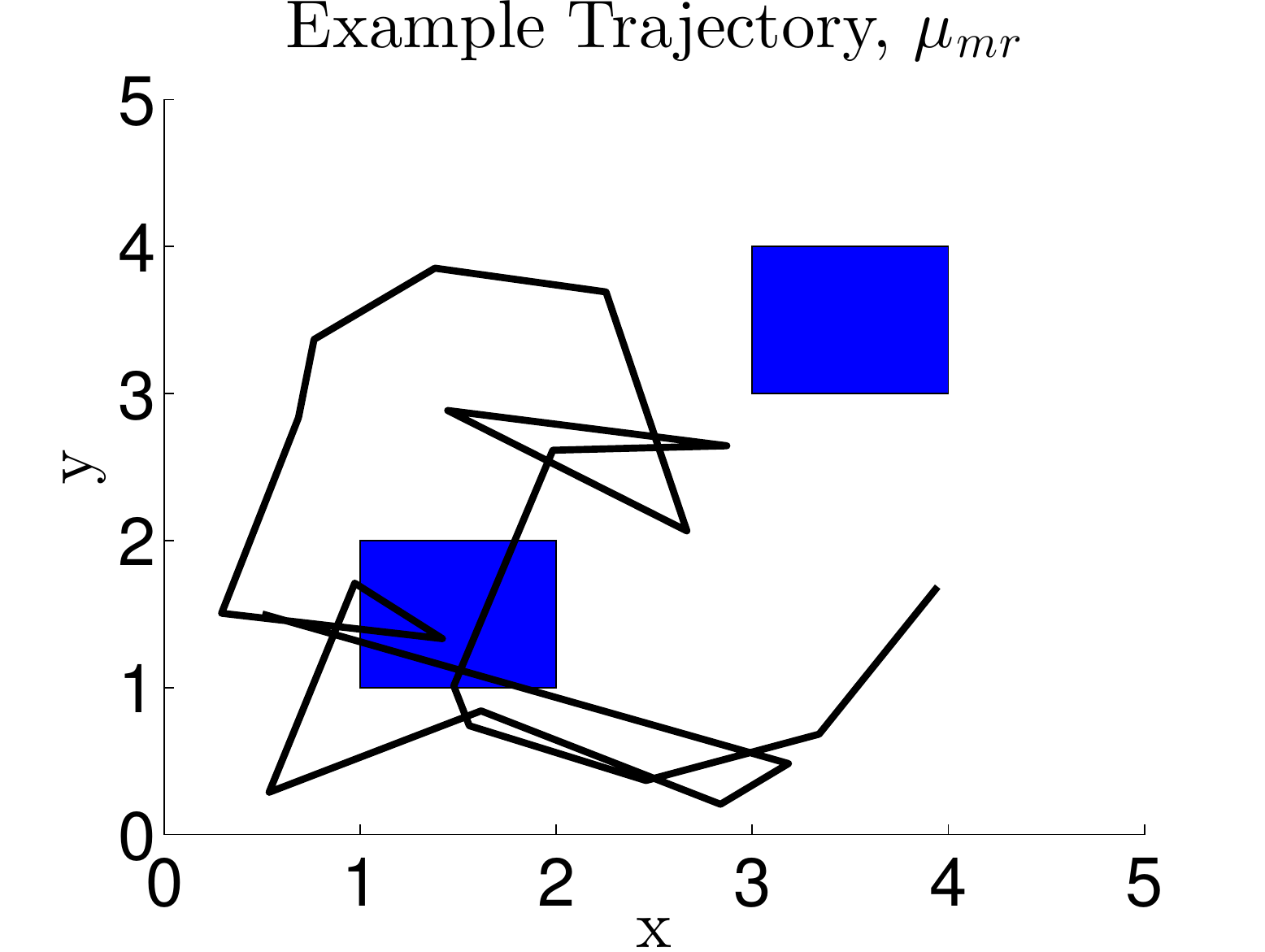} \\
 (c) & (d) \\
\end{tabular}
\caption{\label{caseStudy1} \small Comparison of Policies for Case
  Study 1.
  Histogram of robustness values for trained policies for solution
  to (a) Problem \ref{maxProb} and (b) Problem
  \ref{maxRob}. Trajectory generated from  policies for
  solution to (c) Problem \ref{maxProb} and (d) Problem
  \ref{maxRob}.}
\end{figure}

The two approaches perform very similarly.  In the first row, we show
a histogram of the robustness of 500 trials generated from the system
simulated using each of the trained policies after learning has
completed, i.e. without the randomization that is used during the
learning phase.  Note that both trained policies satisfied the
specification with probability 1.  The performance of the two algorithms 
are very similar, as the mean robustness is  0.2287 
with standard deviation 0.1020 for probability maximization and 0.2617 and 0.1004,resp., for robustness maximization. 
 In the second
row, we see trajectories simulated by each of the trained
policies.

The similarity of the solutions in this case study is not surprising.
If the state of the system is deep within $A$ or $B$, then the
probability that it will remain inside that region in the next 3 time
steps (satisfy $\phi$) is higher than if it is at the edge of the
region. Trajectories that remain deeper in the interior of region $A$
or $B$ also have a high robustness value.  Thus, for this particular
problem, there is an inherent coupling between the policies that
satisfy the formula with high probability and those that satisfy the
formula as robustly as possible on average.
\subsection{Case Study 2: Repeated Satisfaction}

In this second case study, we look at a problem involving repeatedly
satisfying a condition finitely many times.  The specification of
interest is
\begin{equation}
\label{caseStudy2spec}
 \phi_{cs2} = G_{[0,12)} (F_{[0,4)} (\varphi_{blue} ) \wedge  F_{[0,4)} (\varphi_{green}) ),
\end{equation}
  In
plain English, \eqref{caseStudy2spec} is ``Ensure that
every 4 time units over a 12 unit interval, a green region and a blue region is entered.'' Results from this case study are shown in Figure
\ref{caseStudy2}. We
used the same parameters as listed in Section \ref{reachability}, except $N_ep$ = 1200,$\alpha=0.4$, and
$\epsilon^t  = 0.9^t$. Constructing the $\tau$-MDP took 16.5s.  Applying
Algorithm \ref{batchQAlg} took 257.7s for Problem \ref{maxProb} and 258.3s
for Problem \ref{maxRob}.
\begin{figure}
\begin{tabular}{cc}
 \includegraphics[width=0.43\columnwidth]{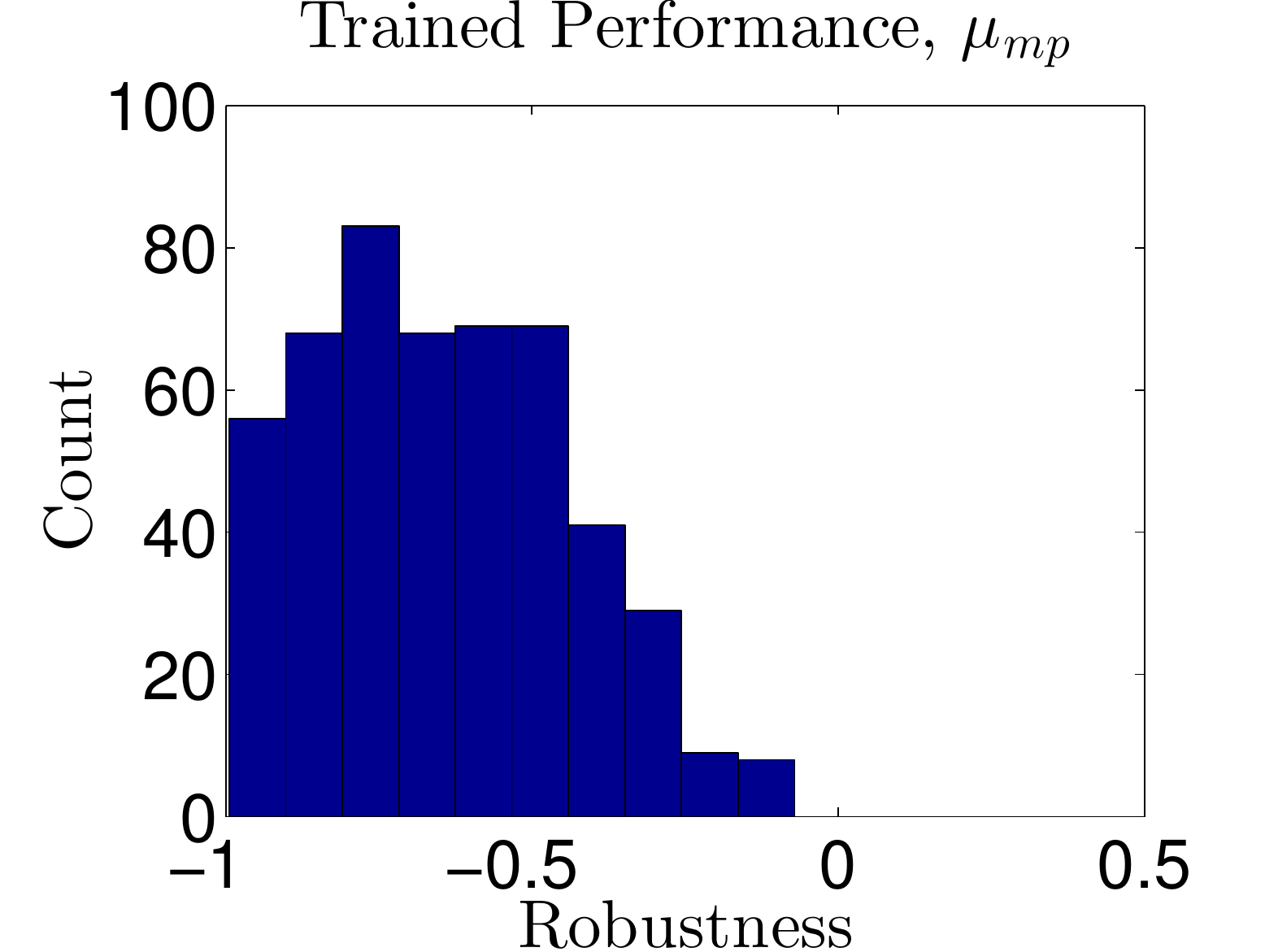} & \includegraphics[width=0.43\columnwidth]{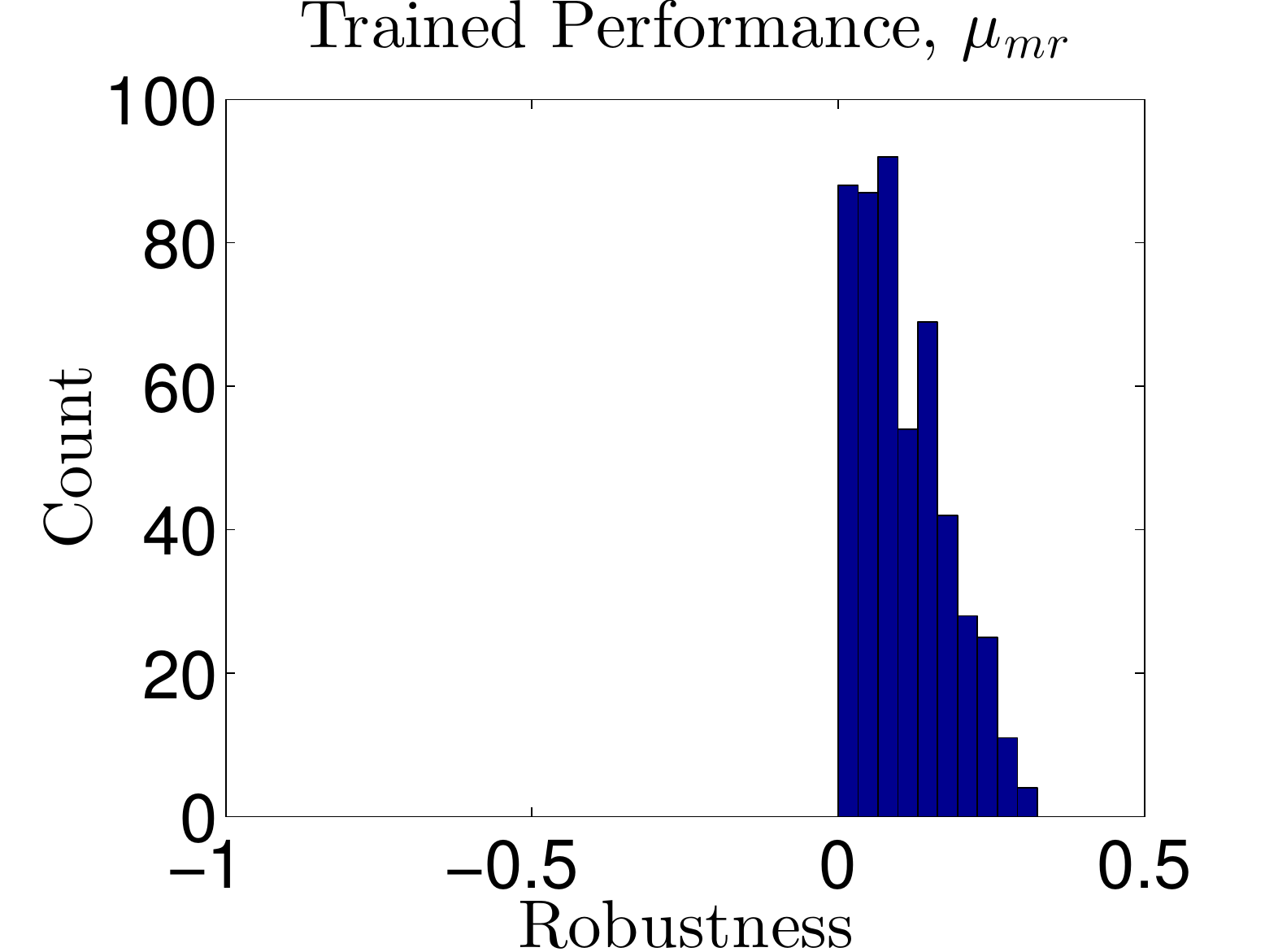} \\
 (a) & (b) \\
 \includegraphics[width=0.43\columnwidth]{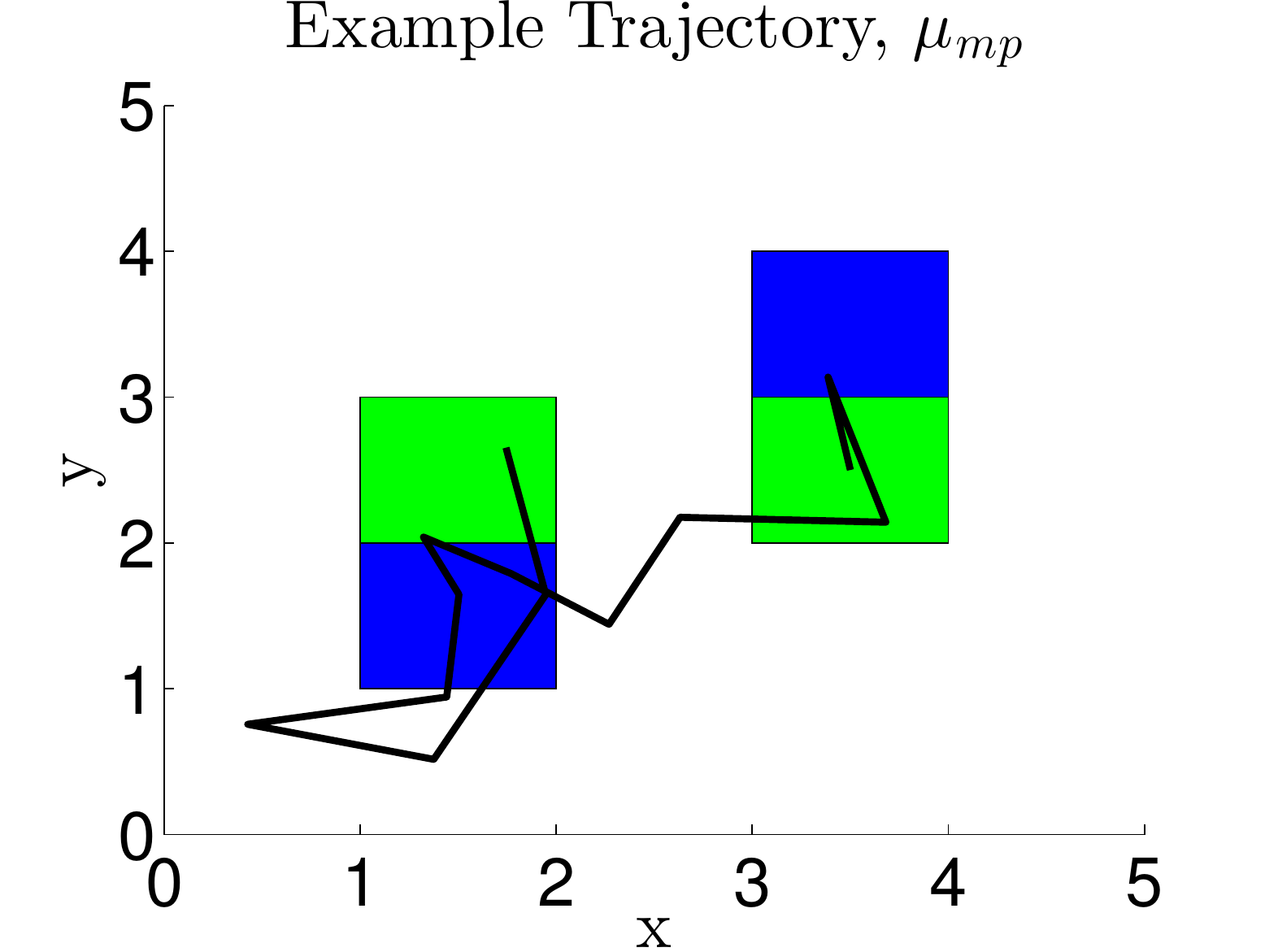} & \includegraphics[width=0.43\columnwidth]{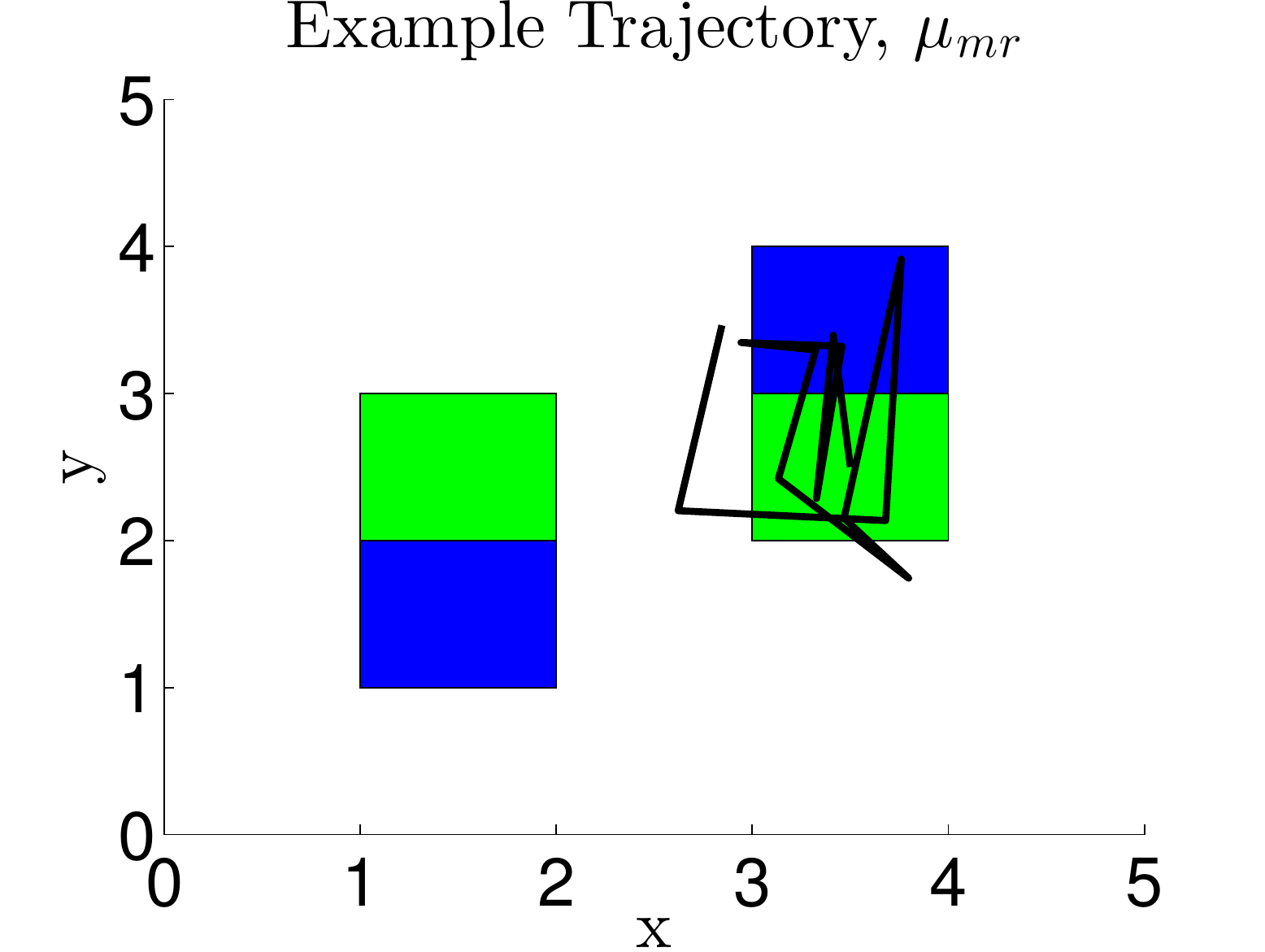} \\
 (c) & (d) \\
\end{tabular}
\caption{\label{caseStudy2} \small Comparison of Policies for Case
  Study 2..
  The subplots have the same meaning as in Figure \ref{caseStudy1}.}
\end{figure}

In the first row, we see that the solution to Problem \ref{maxProb}
satisfies the formula with probability 0 while the solution to Problem
\ref{maxRob} satisfies the formula with probability 1.  At first, this
seems counterintuitive, as Proposition \ref{Qlearnconverge} indicates
that a policy that maximizes probability would achieve a probability
of satisfaction at least as high as the policy that maximizes the
expected robustness. However, this is only guaranteed with an infinite
number of learning trials.  The performance in terms of robustness is obviously better
for the robustness maximization (mean  0.1052, standard deviation
 0.0742) than for the probability maximization (mean  -0.6432, standard
deviation 0.2081). In the second row, we see that the maximum robustness 
policy enforces convergence to a cycle between two regions, while the maximum probability
policy deviates from this cycle.

The discrepancy between the two solutions can be explained by what
happens when trajectories that almost satisfy \eqref{caseStudy2spec}
occur. If a trajectory that almost oscillates between
a blue and green region every four seconds is encountered when solving
Problem \ref{maxProb}, it collects 0 reward.  On the other
hand, when solving Problem \ref{maxRob}, the policy that produces the
almost oscillatory trajectory will be reinforced much more strongly,
as the resulting robustness is less negative.   However, since the robustness degree gives ``partial
credit'' for trajectories that are close to satisfying the policy, the
reinforcement learning algorithm performs a directed search to find
policies that satisfy the formula.  Since probability maximization
gives no partial credit, the reinforcement learning algorithm is
essentially performing a random search until it encounters a
trajectory that satisfies the given formula.  Therefore, if the family
of policies  that satisfy the formula with positive probability is small,
it will on average take the $Q$-learning algorithm solving Problem
\ref{maxProb} a longer time to converge to a solution that enforces
formula satisfaction.

\section{Conclusions and Future Work}
\label{conclusion}
In this paper, we presented a new reinforcement learning paradigm to enforce temporal logic specifications when
the dynamics of the system are \em a priori \em unknown.  In contrast to existing works
on this topic, we use a logic (signal temporal logic) whose formulation
is directly related to a system's statespace. We present a novel, convergent $Q$-learning algorithm that uses the robustness
degree, a continuous measure of how well a trajectory satisfies
a formula, to enforce the given specification.  In certain cases, robustness maximization
subsumes the established paradigm of probability maximization and, in certain cases, 
robustness maximization performs better in terms of both probability and robustness under partial training.
Future research includes formally connecting our approach to abstractions
of linear stochastic systems.

                                  
\bibliographystyle{abbrv}
\bibliography{../references}

\end{document}